\documentclass[a4paper,11pt]{article}

\pdfoutput=1

\usepackage{amsmath,amssymb,mathtools}
\usepackage{array}
\usepackage{bm}
\usepackage{braket}
\usepackage{amsmath, amssymb, amsthm}
\usepackage[table]{xcolor}
\newtheorem{theorem}{Theorem}

\usepackage{booktabs} 
\usepackage{array}    
\usepackage{amsmath}  
\usepackage{xcolor}   
\usepackage{colortbl} 

\usepackage{color}
\usepackage{graphicx}
\usepackage{subfigure}
\usepackage{cite}
\usepackage[colorlinks=true,linkcolor=blue, citecolor=red, urlcolor=blue, bookmarks]{hyperref}
\usepackage{multirow,makecell} 
\usepackage[figuresright]{rotating} 
\usepackage{textcomp}
\usepackage{wasysym}
\usepackage{diagbox}
\usepackage{ulem}

\usepackage[utf8]{inputenc}
\usepackage[T1]{fontenc}

\usepackage[text={17.1cm,24.6cm},centering]{geometry} 

\numberwithin{equation}{section}

\def \be {\begin{equation}}
\def \ee {\end{equation}}
\def \ba {\begin{array}}
\def \ea {\end{array}}
\def \bea{\begin{eqnarray}}
\def \eea{\end{eqnarray}}

\def \S {\Sigma}

\def \tr {\textrm{tr}}

\def \and {{~\textrm{and}~}}

\begin{document}
\sloppy

\newpage
\setcounter{page}{1}

\begin{center}{\Large \textbf{
Matrix Elements of Fermionic Gaussian Operators in Arbitrary Pauli Bases: A Pfaffian Formula\\
}}\end{center}

\begin{center}
M. A. Rajabpour\textsuperscript{1$\star$}
MirAdel Seifi MirJafarlou\textsuperscript{1} and
Reyhaneh Khasseh\textsuperscript{2},

\end{center}

\begin{center}
{\bf 1} Universidade Federal Fluminense, Niter\'oi, Brazil
\\
{\bf 2} Theoretical Physics III, Center for Electronic Correlations and Magnetism, Institute of Physics, University of Augsburg, D-86135 Augsburg, Germany
\end{center}

\begin{center}
\today
\end{center}


\section*{Abstract}
Fermionic Gaussian operators are foundational tools in quantum many-body theory, numerical simulation of fermionic dynamics, and fermionic linear optics. While their structure is fully determined by two-point correlations, evaluating their matrix elements in arbitrary local spin bases remains a nontrivial task, especially in applications involving quantum measurements, tomography, and basis-rotated simulations. In this work, we derive a fully explicit and general Pfaffian formula for the matrix elements of fermionic Gaussian operators between arbitrary Pauli product states. Our approach introduces a pair of sign-encoding matrices whose classification leads to a Lie algebra isomorphic to $\mathfrak{so}(2L)$. This algebraic structure not only guarantees consistency of the Pfaffian signs but also reveals deep connections to Clifford algebras. The resulting framework enables scalable computations across diverse fields---from quantum tomography and entanglement dynamics to algebraic structure in fermionic circuits and matchgate computation. Beyond its practical utility, our construction sheds light on the internal symmetries of Gaussian operators and offers a new lens through which to explore their role in quantum information and computational models.

\vspace{10pt}
\noindent\rule{\textwidth}{1pt}
\tableofcontents
\noindent\rule{\textwidth}{1pt}
\vspace{10pt}

\section{Introduction}\label{sec:Introduction}
Fermionic Gaussian states and operators form a foundational class of quantum objects at the core of non-interacting fermionic systems\cite{Thouless1960,LIEB1961407,Peschel2001,Kitaev2003,TerhalDiVincenzo2002,Zanardi2008,Surace2022,Terhal2023,Peschel2003,PeschelEisler2009}. Entirely characterized by their two-point correlation functions—encoded in covariance matrices—they describe a wide range of physical states, including ground and thermal states of quadratic Hamiltonians\cite{Surace2022,Kraus2010,Finster2009,Hackl2021,Boutin2021,Bravyi2017}, as well as pure states evolving under free-fermion dynamics\cite{Calabrese2007,CEF2011}. Their mathematical structure enables the use of Wick’s theorem\cite{Gaudin1960,BalianBrezin1969,Perk1984}, leading to efficient evaluation of observables \cite{Pfeuty1970,Barouch1971} and entanglement measures \cite{Peschel2001,Kitaev2003} through compact Pfaffian expressions\cite{TKR2024,Rajabpour2025}, see \cite{Surace2022,Hackl2021} for reviews. These properties make them especially valuable in quantum many-body theory and the  simulation of fermionic models.

An interesting development in the classical simulation of fermionic systems is the introduction of matchgate circuits—quantum circuits composed of a special class of two-qubit gates that preserve Gaussianity. Under the Jordan–Wigner transformation, these circuits implement fermionic Gaussian unitaries and can be simulated efficiently on a classical computer under specific conditions \cite{Knill2001,JozsaMiyake2008,CaiGorenstein2014,BravyiKitaev2002,Bravyi2005,Bravyi2012}. The underlying algebraic structure has been shown to correspond to fermionic linear optics, with matchgate circuits generating a group isomorphic to $\mathrm{SO}(2L)$ and Lie algebra $\mathfrak{so}(2L)$. This algebraic structure supports their classical simulability and reveals intrinsic links to Clifford algebras and the geometry of orthogonal transformations \cite{Bravyi2005,JozsaMiyake2008}.
Recent research has expanded the utility of matchgate circuits beyond simulation alone, including inner products with fermionic Gaussian states and expectation values of local fermionic operators~\cite{Gorman2022,Wan2023,Heyraud2025}. 
These advancements have  implications for quantum tomography and the mitigation of the fermion sign problem in quantum Monte Carlo simulations. Furthermore, studies have explored the classical simulation of matchgate circuits augmented with non-Gaussian elements, providing algorithms that retain polynomial complexity under certain conditions.

While these techniques enable efficient numerical estimation of specific observables and overlaps—particularly those compatible with Gaussian structure—they do not offer a general analytical framework for evaluating matrix elements of fermionic Gaussian operators in arbitrary Pauli product bases. In such bases, the loss of manifest Gaussianity complicates direct computation, especially in contexts involving quantum measurements, tomography, and basis-rotated simulations. In this work, we overcome this difficulty by deriving a fully explicit and general Pfaffian formula for such matrix elements.


Traditionally, most of the calculations are carried out in the computational basis, where fermionic modes are associated with the occupation numbers of qubits via the Jordan–Wigner transformation~\cite{JordanWigner1928,LIEB1961407}:
\begin{equation}\label{eq:JW}
    c_l = \prod_{j<l} (-\sigma_j^z)\,\sigma_l^-, 
    \quad
    c_l^\dagger = \prod_{j<l} (-\sigma_j^z)\,\sigma_l^+.
\end{equation}
Since $\sigma_l^z = 2\,c_l^\dagger c_l - 1$, one obtains the representation of any Gaussian operator in the $\{\ket{\uparrow}_z,\ket{\downarrow}_z\}$ basis by the simple substitutions $\ket{0}\mapsto\ket{\downarrow}_z$, $\ket{1}\mapsto\ket{\uparrow}_z$.  However, many physical and experimental protocols—such as randomized measurements ~\cite{Elben2019,Brydges2019,Elben2023}, basis‐rotated tomography \cite{Leone2024}, and post‐measurement entanglement studies~\cite{R2015,Najafi2016,Baweja2024,Rajabpour2025}—require evaluating these operators in arbitrary local spin bases. To do so, one applies a product of single‐site unitaries,
\begin{equation}\label{eq:rot}
    \mathcal{G}^{(\boldsymbol{\phi},\boldsymbol{\theta},\boldsymbol{\alpha})}
    \;=\;
    U_{(\boldsymbol{\phi},\boldsymbol{\theta},\boldsymbol{\alpha})}\,
    \mathcal{G}\,
    U_{(\boldsymbol{\phi},\boldsymbol{\theta},\boldsymbol{\alpha})}^\dagger,
    \quad
    U_{(\boldsymbol{\phi},\boldsymbol{\theta},\boldsymbol{\alpha})}
    = 
    \bigotimes_{l=1}^L
    U_{(\phi_l,\theta_l,\alpha_l)},
\end{equation}
where each
\begin{equation}\label{eq:U}
    U_{(\phi,\theta,\alpha)}
    =
    \begin{pmatrix}
      \cos\tfrac{\theta}{2} & \sin\tfrac{\theta}{2}\,e^{-i\phi} \\
      \sin\tfrac{\theta}{2}\,e^{-i\alpha} & -\cos\tfrac{\theta}{2}\,e^{-i(\alpha+\phi)}
    \end{pmatrix}
\end{equation}
rotates the local $\sigma^z$‐basis into the desired Pauli direction.  While in principle one can expand the rotated operator into Pauli strings, this expansion typically involves exponentially many terms.  Hence, compact and exact Pfaffian formulas for matrix elements in arbitrary spin bases are both highly desirable and nontrivial to derive.

In this work, we derive an explicit Pfaffian formula for the matrix elements of fermionic Gaussian operators in arbitrary local Pauli bases.  Our result allows one to compute 
\begin{equation}
    \langle \mathcal{S} \,|\, \mathcal{G}\,|\,\mathcal{S}'\rangle
\end{equation}
 without ever reverting to the computational basis.  The central technical achievement is a consistent scheme for assigning Pfaffian signs in rotated bases, implemented via a pair of combinatorial sign‐encoding matrices \(\boldsymbol{\Sigma}\) and \(\boldsymbol{\Sigma}'\).

These matrices capture all antisymmetric contractions of fermionic operators and close under a finite group whose Lie algebra is isomorphic to \(\mathfrak{so}(2L)\), the algebra of real orthogonal transformations in \(2L\) dimensions.  This algebraic identification not only guarantees sign‐consistency but also forges a clear bridge between fermionic Gaussian operators, Clifford–Majorana representations, and the geometry of orthogonal groups.

Our framework substantially generalizes prior Pfaffian methods for pure Gaussian states.  In~\cite{TKR2024}, introduced a Pfaffian‐based scheme for amplitudes of Gaussian pure states in the \(\sigma^x\) and \(\sigma^y\) directions via duality and domain‐wall constructions. More recently, ~\cite{Rajabpour2025} gave a fully explicit Pfaffian formula for pure‐state amplitudes in arbitrary Pauli bases and a recursive scaling relation across system sizes.  By contrast, the present work applies to the full class of Gaussian operators—including density matrices of subsystems and unitary evolution operators—and provides a single closed‐form expression for any pair of local spin configurations.

From a computational standpoint, our formula is efficient and scalable: it enables exact and efficient evaluation of formation probabilities~\cite{Shiroishi2001,Franchini2005,Stephan2013,NR2016,Rajabpour2020,Rajabpour2021,MortezaRajab2020}, Shannon–Rényi entropies~\cite{Stephan2009,Stephan2010,Alcaraz2013,Alcaraz2014,NR2016,Tarighi2022,Jiaju2023,central-charge}, post‐measurement updates~\cite{R2015,Najafi2016,Rajabpour2016,Baweja2024,Rajabpour2025}, and other observables in rotated Pauli bases with only \(O(L^3)\) Pfaffian or matrix operations.  Structurally, it deepens the connection between fermionic linear‐optical circuits, Lie algebras, and spin‐basis transformations, and suggests new ansatzes for hybrid spin–fermion simulations and tomography.


The rest of the paper is structured as follows: In Section~\ref{sec:basics}, we introduce fermionic Gaussian operators, define their general form, and discuss relevant algebraic constraints. Section~\ref{sec:computational-basis} presents a closed-form Pfaffian formula for matrix elements in the computational basis using a Grassmann integral approach. While aspects of the theorem presented in this section  may be known in certain restricted cases, to the best of our knowledge, the general formulation presented here has not previously appeared in the literature. In Section~\ref{sec:sigmaz-basis}, we derive an alternative Pfaffian representation in the $\sigma^z$ spin basis and introduce the sign-encoding matrices $\boldsymbol{\Sigma}$ and $\boldsymbol{\Sigma}'$. Section~\ref{sec:arbitrary-bases} generalizes the construction to arbitrary local Pauli bases, providing a fully explicit formula for matrix elements. Section~\ref{sec:Sigma-matrices} explores the combinatorial and algebraic structure of the sign matrices, including their Lie algebra closure and Clifford algebra embedding. We conclude in Section~~\ref{sec:conclusion} with a summary of results and a discussion of further directions.
\section{Fermionic Gaussian Operators: Definitions and Basics}\label{sec:basics}
In this section, we introduce the basic formalism of fermionic Gaussian operators and establish the notations used throughout the paper. In physical contexts, these operators are either Hermitian, when dealing with mixed Gaussian states, or unitary, when considering the evolution of a system under free fermionic Hamiltonians. Although Gaussian operators can be defined abstractly via the Wick-Gaudin theorem~\cite{Gaudin1960,Perk1984}—particularly for mixed states—their explicit structure is not always readily apparent. For simplicity, we will adopt the definition of a Gaussian operator as an explicit exponential with a quadratic expression in terms of creation and annihilation operators as its argument. Based on this reasoning, we define the fermionic Gaussian operator as follows:
\begin{equation}\label{Generic Gaussian operator}
   \mathcal{G}_{\bm{\mathcal{M}}} = \exp\!\Biggl[\frac{1}{2}\begin{pmatrix} \mathbf{c}^{\dagger} & \mathbf{c} \end{pmatrix} \bm{\mathcal{M}} \begin{pmatrix} \mathbf{c} \\ \mathbf{c}^{\dagger} \end{pmatrix}\Biggr],
\end{equation}
where $\bigl(\mathbf{c}^{\dagger}, \mathbf{c}\bigr) = \Bigl(c_1^{\dagger}, c_2^{\dagger}, \dots, c_L^{\dagger}, c_1, c_2, \dots, c_L\Bigr)$ and \(\bm{\mathcal{M}}\) is a \(2L \times 2L\) matrix. Without loss of generality, we assume that:
\begin{equation}
\boldsymbol{\Xi}\bm{\mathcal{M}} + \bigl(\boldsymbol{\Xi}\bm{\mathcal{M}}\bigr)^T = 0,
\end{equation}
with
\begin{equation}
\boldsymbol{\Xi} = \begin{pmatrix} \mathbf{0} & \mathbf{I} \\ \mathbf{I} & \mathbf{0} \end{pmatrix}.
\end{equation}
A number‑conserving special case of the Gaussian operator can be written as follows:
\begin{equation}\label{special Gaussian operator}
   \mathcal{G}_{\mathbf{A}} = \exp\!\Big[\mathbf{c}^{\dagger} \mathbf{A} \mathbf{c}\Big],
\end{equation}
where $\mathbf{A}$ can be any $L\times L$ general matrix. When working with Gaussian mixed states, it is necessary to impose the additional conditions of Hermiticity and normalization. In such cases, the following holds:
\begin{equation}\label{GMS fermionic operator}
   \boldsymbol{\rho}_{\bm{\mathcal{M}}} = \frac{1}{Z_{\bm{\mathcal{M}}}} \exp\!\Biggl[\frac{1}{2}\begin{pmatrix} \mathbf{c}^{\dagger} & \mathbf{c} \end{pmatrix} \bm{\mathcal{M}} \begin{pmatrix} \mathbf{c} \\ \mathbf{c}^{\dagger} \end{pmatrix}\Biggr],
\end{equation}
where
\begin{equation}
\boldsymbol{\Xi}\bm{\mathcal{M}} + \bigl(\boldsymbol{\Xi}\bm{\mathcal{M}}\bigr)^T = 0 \quad \text{and} \quad \bm{\mathcal{M}}^{\dagger} = \bm{\mathcal{M}}.
\end{equation}
The normalization factor \(Z_{\bm{\mathcal{M}}}\) is determined by
\begin{equation}
\begin{split}
Z_{\bm{\mathcal{M}}} &= \text{tr}\,\exp\!\Biggl[\frac{1}{2}\begin{pmatrix} \mathbf{c}^{\dagger} & \mathbf{c} \end{pmatrix} \bm{\mathcal{M} }\begin{pmatrix} \mathbf{c} \\ \mathbf{c}^{\dagger} \end{pmatrix}\Biggr] \\
&\;=\det\!\left(\mathbf{I}+e^{\bm{\mathcal{M}}}\right)^{\frac{1}{2}}.
\end{split}
\end{equation}
In the next section, we derive explicit Pfaffian expressions for matrix elements of Gaussian operators in the computational (i.e., $\sigma^z$) basis.
\section{Matrix Elements in the Computational Basis}\label{sec:computational-basis}
In this section, we derive explicit formulas for the matrix elements of Gaussian operators in the computational basis. To achieve this, we first express the operator in the fermionic coherent state basis. By setting certain Grassmann variables to zero and integrating over the remaining ones, we obtain an explicit Pfaffian formula. This Pfaffian formula provides all the matrix elements of the Gaussian operator.

\begin{theorem}[Matrix Elements of Fermionic Gaussian Operators in the Computational Basis]
Let $\mathcal{G}_{\bm{\mathcal{M}}}$ be a fermionic Gaussian operator acting on an $L$-mode fermionic Fock space, defined as
\begin{equation}\label{eq:gaussian_operator}
\mathcal{G}_{\bm{\mathcal{M}}} = \exp\!\Biggl[\frac{1}{2}\begin{pmatrix} \mathbf{c}^{\dagger} & \mathbf{c} \end{pmatrix} \bm{\mathcal{M}} \begin{pmatrix} \mathbf{c} \\ \mathbf{c}^{\dagger} \end{pmatrix}\Biggr],
\end{equation}
with $\mathbf{c} = (c_1, c_2, \dots, c_L)^T$ and $\bm{\mathcal{M}}$ a $2L\times 2L$ matrix satisfying the appropriate symmetry conditions. Assume that the exponential $e^{\bm{\mathcal{M}}}$ admits the block decomposition
\begin{equation}\label{eq:matrixM}
e^{\bm{\mathcal{M}}} = \begin{pmatrix} \mathbf{T}_{11} & \mathbf{T}_{12} \\ \mathbf{T}_{21} & \mathbf{T}_{22} \end{pmatrix},
\end{equation}
and define
\begin{equation}\label{eq:matricesXYZ}
\mathbf{X} = \mathbf{T}_{12} (\mathbf{T}_{22})^{-1}, \quad \mathbf{Z} = (\mathbf{T}_{22})^{-1}\mathbf{T}_{21}, \quad e^{-\mathbf{Y}} = \mathbf{T}_{22}^T.
\end{equation}
For any subset $\mathcal{I}\subset\{1,2,\dots,L\}$, define the computational basis state
\begin{equation}\label{eq:computationalI}
\ket{\mathcal{I}} = c^\dagger_{i_1} c^\dagger_{i_2}\cdots c^\dagger_{i_p}\ket{0}, \quad \text{with } \mathcal{I}=\{i_1,i_2,\dots,i_p\}\ (i_1 < i_2 < \cdots < i_p),
\end{equation}
and for another subset $\mathcal{J}\subset\{1,2,\dots,L\}$ define the dual state
\begin{equation}\label{eq:computationalII}
\bra{\mathcal{J}} = \bra{0}\, c_{j_p}\cdots c_{j_2} c_{j_1}, \quad \text{with } \mathcal{J}=\{j_1,j_2,\dots,j_p\}\ (j_1 < j_2 < \cdots < j_p).
\end{equation}
Introduce the partitions
\begin{equation}\label{eq:partitions}
\mathcal{I}_1 \equiv \mathcal{I}, \quad \mathcal{I}_0 \equiv \{1,2,\dots,L\}\setminus \mathcal{I}, \quad \mathcal{J}_1 \equiv \mathcal{J}, \quad \mathcal{J}_0 \equiv \{1,2,\dots,L\}\setminus \mathcal{J}.
\end{equation}
Then, the matrix element of $\mathcal{G}_{\bm{\mathcal{M}}}$ between the computational basis states is given by:

\begin{equation}\label{eq:matrix_element-computational}
\langle \mathcal{J}|\mathcal{G}_{\bm{\mathcal{M}}}|\mathcal{I} \rangle = (-1)^{\frac{|\mathcal{I}_1|(|\mathcal{I}_1|+2|\mathcal{J}_1|+1)}{2}}\, \operatorname{det}\Bigl[T_{22}\Bigr]^{\frac{1}{2}}\operatorname{pf}\!\Bigl[\bm{\mathcal{A}}_{\mathcal{J}_0\mathcal{I}_0}\Bigr],
\end{equation}
where the $2L\times 2L$ matrix $\bm{\mathcal{A}}$ is defined as
\begin{equation}\label{eq:matrixA}
\bm{\mathcal{A}} = \begin{pmatrix} \mathbf{X} & e^{\mathbf{Y}} \\ -e^{\mathbf{Y}^T} & \mathbf{Z} \end{pmatrix},
\end{equation}
and $\bm{\mathcal{A}}_{\mathcal{J}_0\mathcal{I}_0}$ denotes the submatrix of $\bm{\mathcal{A}}$ obtained by deleting the rows corresponding to indices in $\mathcal{J}_0$ and the rows indexed by $L+\mathcal{I}_0$, and analogously for the columns.
\end{theorem}

\begin{proof}[Proof]The proof proceeds in the following three steps.

\paragraph{Step 1: Balian--Brezin Decomposition.} Begin with the Balian--Brezin decomposition \cite{BalianBrezin1969} of $\mathcal{G}_{\bm{\mathcal{M}}}$:
\begin{equation}\label{eq:BB_decomp}
\mathcal{G}_\mathcal{M} = \exp\!\Bigl[\frac{1}{2}\,\mathbf{c}^{\dagger}\,\mathbf{X}\,\mathbf{c}^{\dagger}\Bigr]\, \exp\!\Bigl[\mathbf{c}^{\dagger}\,\mathbf{Y}\,\mathbf{c} - \frac{1}{2}\operatorname{Tr}\mathbf{Y}\Bigr]\, \exp\!\Bigl[\frac{1}{2}\,\mathbf{c}\,\mathbf{Z}\,\mathbf{c}\Bigr],
\end{equation}
with
\[
\mathbf{X} = \mathbf{T}_{12} (\mathbf{T}_{22})^{-1}, \quad \mathbf{Z} = (\mathbf{T}_{22})^{-1}\mathbf{T}_{21}, \quad e^{-\mathbf{Y}} = \mathbf{T}_{22}^T,
\]
which follows from the block decomposition of $e^{\bm{\mathcal{M}}}$. Note that $\mathbf{X}$ and $\mathbf{Z}$ are antisymmetric matrices. To compute the matrix elements in the computational basis, express $\mathcal{G}_{\bm{\mathcal{M}}}$ in the fermionic coherent state basis $\{\ket{\boldsymbol{\xi}}\}$, where the creation and annihilation operators are replaced by Grassmann variables.

\bigskip

\paragraph{Step 2: Coherent State Representation of Basis States.} In the coherent state representation, the computational basis states in the fermionic Fock space are expressed via Berezin integrals:
\begin{eqnarray}\label{eq:coherentstateI-J}
\ket{\mathcal{I}} &=& (-1)^{\frac{|\mathcal{I}_1|(|\mathcal{I}_1|+1)}{2}} \int \prod_{i\in \mathcal{I}_1} d\xi'_i \; \ket{\boldsymbol{\xi}'(\mathcal{I})},\\
\bra{\mathcal{J}} &=& \int \prod_{j\in \mathcal{J}_1} d\xi_j \; \bra{\boldsymbol{\xi}(\mathcal{J})},
\end{eqnarray}
where in $\ket{\boldsymbol{\xi}'(\mathcal{I})}$ the Grassmann variables corresponding to the unoccupied sites $\mathcal{I}_0 = \{1,2,\dots,L\}\setminus\mathcal{I}$ are set to zero, and similarly for $\bra{\boldsymbol{\xi}(\mathcal{J})}$.

\bigskip

\paragraph{Step 3: Berezin Integration and Pfaffian Form.} Expressing $\mathcal{G}_{\bm{\mathcal{M}}}$ in the coherent state basis leads to an integral over Grassmann variables. Performing the Berezin integrations over the occupied sites (with the variables corresponding to unoccupied sites set to zero) yields an expression in which the contributions combine into a Pfaffian. The required reordering of the Grassmann variables introduces an overall sign factor
\[
(-1)^{\frac{|\mathcal{I}_1|(|\mathcal{I}_1|+2|\mathcal{J}_1|+1)}{2}}.
\]
Moreover, the structure of the integrals produces the Pfaffian of a submatrix of the $2L\times 2L$ matrix \bm{\mathcal{A}}, defined in Eq.~\eqref{eq:matrixA}. The submatrix $\bm{\mathcal{A}}_{\mathcal{I}_0\mathcal{J}_0}$ is obtained by deleting the rows corresponding to indices in $\mathcal{J}_0$ and the rows corresponding to $L+\mathcal{I}_0$, and analogously for the columns.

\end{proof}

\subsection{Particle-Preserving Gaussian Operators}
When the Gaussian operator is particle preserving, i.e. (\ref{special Gaussian operator}), then the equations simplify significantly. We have
\begin{equation}\label{density matrix computational basis}
\bra{\mathcal J}\,\mathcal G_{\mathbf A}\,\ket{\mathcal I}
=
\begin{cases}
\displaystyle
\det\!\bigl[(e^{A})_{\mathcal J_0,\mathcal I_0}\bigr],
&|\mathcal I_0|=|\mathcal J_0|,\\[8pt]
0,&\text{otherwise.}
\end{cases}
\end{equation}
Here, $|\mathcal I_0|$ (resp.\ $|\mathcal J_0|$) denotes the number of unoccupied sites in the basis $\ket{\mathcal I}$ (resp.\ $\bra{\mathcal J}$), and $(e^{A})_{\mathcal J_0,\mathcal I_0}$ is the submatrix (minor) of $e^{A}$ obtained by deleting rows $\mathcal J_0$ and columns $\mathcal I_0$.
\subsection{Gaussian Mixed State}
When dealing with a Gaussian mixed state, we have
\begin{eqnarray}
\mathbf{X}=\mathbf{Z}^{\dagger},\hspace{1cm} \mathbf{Y}^{\dagger}=\mathbf{Y}.
\end{eqnarray}
Note that for diagonal elements the extra sign in the equation (\ref{density matrix computational basis}) disappears. When all the elements of the matrix $\mathcal{M}$ are real, further simplifications occur. For this case the matrix \(\bm{\mathcal{A}}\) can be written as
\begin{equation}\label{eq:Amatrixreal}
\bm{\mathcal{A}}= \begin{pmatrix} \mathbf{F}_a & \mathbf{F}_s \\ -\mathbf{F}_s & -\mathbf{F}_a \end{pmatrix},
\end{equation}
with \(\mathbf{F}_s\) and \(\mathbf{F}_a\) being the symmetric and antisymmetric parts of the matrix \(\mathbf{F}\) (which itself can be expressed in terms of the correlation matrix \(\mathbf{G}\) by \(\mathbf{F} = (\mathbf{I}+\mathbf{G}).({\mathbf{I}-\mathbf{G}})^{-1}\)). The correlation matrix itself is defined as
\begin{equation}\label{G}
G_{jk}=\tr[\rho_{\mathcal{M}} (c_j^{\dagger}-c_j)(c_k^{\dagger}+c_k)].    
\end{equation}
Then one can  write the diagonal elements of the mixed state as \cite{MortezaRajab2020}
\begin{equation}{\label{Prozreal}}
    P_{\mathcal{I}} = \det[\frac{\bold{I}-\bold{I}_{\mathcal{I}}\cdot\bold{G}}{2}],
\end{equation}
where the matrix $\bold{I}_{\mathcal{I}}$ is diagonal, composed of $\pm1$. We assign a diagonal element of $-1$ in cases where a fermion is present, and $1$ in instances where there is an absence of a fermion at the relevant site. The off-diagonal elements of the mixed state in this case do not appear to have a straightforwards relation with the correlation matrix $\bold{G}$. 

\section{Pfaffian Form of Gaussian Operators in the $\sigma^z$ Spin Basis}\label{sec:sigmaz-basis}

In this section, we present an alternative expression for the Gaussian operators in the  $\sigma^z$ basis. This formulation lays the groundwork for deriving the most general expression of Gaussian operators in arbitrary Pauli bases. Here, we denote the basis states by
\begin{equation}
\bra{\mathcal{S}} = \bra{s_L, \ldots, s_2, s_1} \quad \text{and} \quad \ket{\mathcal{S}'} = \ket{s'_{1}, s'_{2}, \ldots, s'_{L}}\equiv \ket{s_{L+1}, s_{L+2}, \ldots, s_{2L}},
\end{equation}
where \(s_j\) and \(s'_j\) are expressed in the $\sigma^z$ basis. Note that we associate the value \(\bar{s} = +1\) and \(\bar{s}' = +1\) to spin up, and \(\bar{s} = -1\) and \(\bar{s}' = -1\) to spin down.

\subsection{Matrix Element Representation}

\begin{theorem}[Gaussian Operators  in the $\sigma^z$ basis]
A generic element of a Gaussian operator in $\sigma^z$ basis can be written as:
\begin{equation}\label{eq:matrix_element-sigmaz}
 \bra{\mathcal{S}} \mathcal{G}_{\bm{\mathcal{M}}} \ket{\mathcal{S'}} =   \operatorname{det}\Bigl[T_{22}\Bigr]^{\frac{1}{2}}\bold{pf}\Bigl[\bm{\mathcal{K}}^z(\mathcal{S},\mathcal{S}')\Bigr].
\end{equation}
There exist $2^{2L-1}$ pairs of  \(2L \times 2L\) antisymmetric matrices $\boldsymbol{\Sigma}$ and $\boldsymbol{\Sigma}'$ with off-diagonal elements made of $\pm1$ such that for indices \(n > m\), the elements of the antisymmetric matrix \(\bm{\mathcal{K}}^z(\mathcal{S},\mathcal{S}')\) can be found as:
\begin{equation}\label{eq:M-sigmaz}
 \mathcal{K}^z_{mn}(\mathcal{S},\mathcal{S}') =(\frac{1+\bar{s}_m}{2})(\frac{1+\bar{s}_n}{2}) \Sigma_{m,n}\,\mathcal{A}_{mn}\, 
 +(\frac{1-\bar{s}_m}{2})(\frac{1-\bar{s}_n}{2})  \Sigma'_{m,n}.
\end{equation}
\end{theorem}

\begin{proof}[Proof]
The core idea of the proof is to ensure that Equation (\ref{eq:matrix_element-sigmaz}) matches Equation (\ref{eq:matrix_element-computational}). This can be achieved by explicitly examining the matrix elements of the Gaussian operator and determining the matrices $\boldsymbol{\Sigma}$ and $\boldsymbol{\Sigma}'$ such that both equations yield identical results.

\paragraph{Case 1: $\bra{\mathcal{S}}=\bra{++...+}$ and $\ket{\mathcal{S'}}= \ket{++...+}$.}
In this case for \(n > m\) we have $\mathcal{B}_{mn}\equiv \mathcal{K}^z_{mn}(\mathcal{S},\mathcal{S}')=\Sigma_{m,n}\,\mathcal{A}_{mn}$ which its pfaffian should match with $(-1)^{\frac{L(3L+1)}{2}}\bold{pf}[\bm{\mathcal{A}}]$. Recall that the Pfaffian of an antisymmetric matrix \( \bold{A} \) is given by
\begin{equation}\label{eq:pfaffian}
\bold{pf}[\bold{A}]= \sum_{M} \epsilon(M) \prod_{(i,j)\in M} A_{ij},
\end{equation}
where the sum runs over all perfect matchings \( M \) of the set \(\{1,2,\dots,2L\}\) and \(\epsilon(M)\) is a sign factor determined by the matching order.
When we form \( \bm{\mathcal{B}} \) by multiplying the corresponding elements of \( \bm{\mathcal{A}} \) by \(\Sigma_{mn}\), its Pfaffian becomes
\[
\bold{pf}[\bm{\mathcal{B}}]= \sum_{M} \epsilon(M) \prod_{(m,n)\in M} \Sigma_{mn}\, \mathcal{A}_{mn}.
\]
In order for \(\bold{pf}[\bm{\mathcal{B}}]\) to equal \(\bold{pf}[\bm{\mathcal{A}}]\) for $L \text{mod} 4=0,1$ an arbitrary antisymmetric matrix \( \bm{\mathcal{A}} \), the extra sign factors must cancel out for every perfect matching \( M \); that is, we require
\begin{equation}\label{eq:perfectmatchingcondition}
\prod_{(i,j)\in M} \Sigma_{ij} = 1 \quad \text{for every perfect matching } M.
\end{equation}
A sufficient and necessary condition to guarantee this is to express the elements of \(\boldsymbol{\Sigma}\) in the form
\begin{equation}\label{eq:sigmavrtop}
\Sigma_{ij} = p_i\,p_j \quad \text{for } i < j,
\end{equation}
where each \( p_i \in \{+1, -1\} \).

Now, consider any perfect matching \( M \). In such a matching, every index \(1,2,\dots,2L\) appears exactly once. Therefore, the product over the matching is
\begin{equation}
\prod_{(i,j)\in M} \Sigma_{ij} = \prod_{(i,j)\in M} (p_i\,p_j) = \prod_{i=1}^{2L} p_i.
\end{equation}
Thus, for this product to be equal to 1 for every perfect matching, we need to have
\begin{equation}\label{eq:pcondition}
\prod_{i=1}^{2L} p_i = 1.
\end{equation}
Since each \( p_i \) can be independently chosen as \( \pm 1 \), there are \( 2^{2L} \) possible choices for the sequence \( \{p_1, p_2, \dots, p_{2L}\} \). However, the constraint \( \prod_{i=1}^{2L} p_i = 1 \) reduces the number of valid solutions by a factor of 2, i.e. $2^{2L-1}$. However changing the sign of all $p_i$s does not change the matrix $\boldsymbol{\Sigma}$ so we end up to $2^{2L-2}$ solutions. Finally we note that if $\boldsymbol{\Sigma}$ is a solution then $-\boldsymbol{\Sigma}$ should be also a solution which 
yielding exactly
\[
2^{2L-1}
\]
possible matrices \(\boldsymbol{\Sigma}\) that satisfy the condition \(\bold{pf}[\mathcal{B}] = \bold{pf}[\mathcal{A}]\).

 When $L \text{mod}4=2,3$ we should have \(\bold{pf}[\mathcal{B}]=-\bold{pf}[\mathcal{A}]\). Here similar argument works with the condition  
\begin{equation}\label{eq:pcondition2}
\prod_{i=1}^{2L} p_i = -1.
\end{equation}
Again there are $2^{2L-2}$ solutions that can be parametrized as $\Sigma_{ij} = p_i\,p_j \quad \text{for } i < j$. The other $2^{2L-2}$ solutions are just the negative of the above solutions. 

\paragraph{Case 2: $\bra{\mathcal{S}}=\bra{++..-+...+-,+,...+}$ and $\ket{\mathcal{S'}}= \ket{++...+}$}: In this case two of the spins in the configuration $\mathcal{S}$ are down spins. Here we do exactly the same analysis: Equation (\ref{eq:matrix_element-sigmaz}) should match Equation (\ref{eq:matrix_element-computational}). This is possible if

\begin{equation}\label{sigmaprime1}
\Sigma'_{nm}=(-1)^{n+m+1}\Sigma_{nm} ,\hspace{1cm} L\geq m>n.
\end{equation}

\paragraph{Case 3: $\bra{\mathcal{S}}=\bra{++...+-+...+}$ and $\ket{\mathcal{S'}}= \ket{++...+-+...+}$}: In this case one of the spins in each of the configurations $\mathcal{S}$ and $\mathcal{S'}$ are down spins. Exactly the same analysis yields

\begin{equation}\label{sigmaprime2}
\Sigma'_{nm}=(-1)^{L+1}(-1)^{n+m+1}\Sigma_{nm} ,\hspace{1cm}  m>L\geq n.
\end{equation}

\paragraph{Case 4: $\bra{\mathcal{S}}=\bra{++...++}$ and $\ket{\mathcal{S'}}= \ket{++-...+-+...+}$}: In this case two of the spins in the configuration $\mathcal{S}'$  are down spins. Here we have

\begin{equation}\label{sigmaprime3}
\Sigma'_{nm}=-(-1)^{n+m+1}\Sigma_{nm} ,\hspace{1cm}  m>n> L.
\end{equation}

The three equation (\ref{sigmaprime1}), (\ref{sigmaprime2}) and (\ref{sigmaprime3}) determine the $\boldsymbol{\Sigma}'$ matrices with respect to the $\boldsymbol{\Sigma}$ matrix. Since there are $2^{2L-1}$ possible $\boldsymbol{\Sigma}$ matrices one can derive exactly $2^{2L-1}$ number of $\boldsymbol{\Sigma}'$ matrices.

\paragraph{General Case:} 
For general configurations where $\bra{\mathcal{S}}$ has $L-|\mathcal{J}_1|$ number of down spins and $\ket{\mathcal{S'}}$ has $L-|\mathcal{I}_1|$ number of down spins, similar calculation leads to the equality
\begin{equation}\label{sigmaprimeGeneral}
(-1)^{\frac{L(3L+1)}{2}}\prod_{j=1}^{L-\frac{|\mathcal{I}_1|+|\mathcal{J}_1|}{2}}(-1)^{m_j+n_j+1}\Sigma'_{m_jn_j}\Sigma_{m_jn_j}=(-1)^{\frac{|\mathcal{I}_1|(|\mathcal{I}_1|+2|\mathcal{J}_1|+1)}{2}}
\end{equation}
In this expression, the index \( j \) runs over the pairs \( (m_j, n_j) \) forming a perfect matching on the set of unoccupied modes \(\{ \mathcal{I}_0 \cup \mathcal{J}_0 \}\). Remarkably the above equation is consistent with the equations (\ref{sigmaprime1}), (\ref{sigmaprime2}), and (\ref{sigmaprime3}) without putting any extra constraints on the matrices. In other words we found $2^{2L-1}$ pairs of solutions for the $\boldsymbol{\Sigma}$ and $\boldsymbol{\Sigma}'$ matrices.
\end{proof}

We now provide one representation for the $\boldsymbol{\Sigma}$ and $\boldsymbol{\Sigma}'$ matrices. In the Appendix (\ref{sec:AppendixA}) all the possible pairs of $\boldsymbol{\Sigma}$ and $\boldsymbol{\Sigma}'$ matrices for $L=2$ and $3$ are listed.
The entry \(\boldsymbol{\Sigma}_{m,n}\) for \(m<n\) is given by the sign rules in Table~\ref{tab:sigma-rules}.
\subsection{Sign Rules for the Matrices \texorpdfstring{\(\boldsymbol{\Sigma}\)}{Sigma} and \texorpdfstring{\(\boldsymbol{\Sigma}'\)}{Sigma'}}

The signs of the upper-triangular entries of the matrix \(\boldsymbol{\Sigma}\) depend on the value of \(L \bmod 4\), as summarized in Table~\ref{tab:sigma-rules}.
\begin{table}[h]
\centering
\begin{tabular}{c c}
  \toprule
  \rowcolor{gray!20}
  \(L \bmod 4\)
    & \(\boldsymbol{\Sigma}_{m,n}\,(m<n)\)
 \\
  \midrule
  0,\,1 
    & \(1\) \\[1ex]
  2 
    & 
    \(
      \begin{cases}
        1  & \text{if } m = 1,\\
       -1  & \text{if } m \neq 1
      \end{cases}
    \) \\[1ex]
  3 
    & \(-1\) \\
  \bottomrule
\end{tabular}
\caption{The upper‐triangular entries of \(\boldsymbol{\Sigma}\) for different values of \(L \bmod 4\).}
\label{tab:sigma-rules}
\end{table}
\noindent Similarly, the sign of each upper-triangular entry in $\boldsymbol{\Sigma}'$ is determined by $L \bmod 4$, with one set of rules for the first superdiagonal and a separate set for all higher superdiagonals (see Table~\ref{tab:sigmaprime-rules}).

%
\begin{table}[h]
\centering
\begin{tabular}{>{\boldmath}c<{} c c}
\toprule
\rowcolor{gray!20}
\textbf{\(L \mod 4\)} & \textbf{ \(\boldsymbol{\Sigma}'_{m,m+1}\)} & \textbf{ \(\boldsymbol{\Sigma}'_{m,m+i}\)(\(i > 1\))} \\
\midrule
0 & 
    \(\begin{cases} 
        1, & 1 \leq m < L \\ 
        -1, & L \leq m < 2L 
    \end{cases}\) & 
    \(\boldsymbol{\Sigma}'_{m,m+i} = -\boldsymbol{\Sigma}'_{m+1,m+i}\) \\
1 & 
    \(\begin{cases} 
        1, & 1 \leq m \leq L \\ 
        -1, & L < m < 2L 
    \end{cases}\) & 
    \(\boldsymbol{\Sigma}'_{m,m+i} = -\boldsymbol{\Sigma}'_{m,m+i-1}\) \\
2 & 
    \(\begin{cases} 
        1, & m = 1 \\ 
        -1, & 1 < m < L \\ 
        1, & L \leq m < 2L 
    \end{cases}\) & 
    \(\begin{cases} 
        \boldsymbol{\Sigma}'_{m,m+i} = \boldsymbol{\Sigma}'_{m+1,m+i}, & m=1 \\ 
        \boldsymbol{\Sigma}'_{m,m+i} = -\boldsymbol{\Sigma}'_{m+1,m+i}, & m>1 
    \end{cases}\) \\
3 & 
    \(\begin{cases} 
        -1, & 1 \leq m \leq L \\ 
        1, & L < m < 2L 
    \end{cases}\) & 
    \(\boldsymbol{\Sigma}'_{m,m+i} = -\boldsymbol{\Sigma}'_{m,m+i-1}\) \\
\bottomrule
\end{tabular}
\caption{The upper‐triangular entries of \(\boldsymbol{\Sigma}^\prime\) for different values of \(L \bmod 4\).}
\label{tab:sigmaprime-rules}
\end{table}
\newpage
\subsection{Generating Function Formulation}
\begin{theorem}[Generating Function for Fermionic Gaussian Operators in the $\sigma^z$ Basis]
\label{thm:gen-func-sigmaz}
Let $\mathcal{G}_{\bm{\mathcal{M}}}$ be a fermionic Gaussian operator specified by a $2L\times2L$ antisymmetric matrix $\bm{\mathcal{A}}$. One can define its generating function in the computational basis by
\begin{equation}\label{generating function defenition}
g(\lambda_1,\dots,\lambda_{2L})
=\sum_{\mathcal{J},\mathcal{I}}
\bigl\langle \mathcal{J}\big|\mathcal{G}_{\bm{\mathcal{M}}}\big|\mathcal{I}\bigr\rangle
\prod_{j\in\mathcal{J}_0\cup\mathcal{I}_0}\lambda_j,
\end{equation}
where 
\[
\mathcal{I}_0 \;=\;\bigl\{\,i\in\mathcal{I}:\text{mode }i\text{ is unoccupied}\bigr\},
\]
and $\mathcal{J}_1,\mathcal{I}_1$ are their complements. Then this generating function admits the compact closed‐form
\begin{equation}\label{generating function defenition final}
g(\lambda_1,\dots,\lambda_{2L})
=\operatorname{\mathbf{pf}}\!\bigl[\bm{\mathcal{K}}^z(\bm{\lambda})\bigr],
\end{equation}
where $\bm{\mathcal{K}}^z(\bm{\lambda})$ is the $2L\times2L$ antisymmetric matrix
\begin{equation}
\label{eq:K-lambda}
\mathcal{K}^z_{mn}(\bm{\lambda})
=\Sigma_{mn}\,\mathcal{A}_{mn}
\;+\;
\Sigma'_{mn}\,\lambda_m\,\lambda_n~.
\end{equation}
\end{theorem}

\begin{proof}
  Starting from the definition of $g(\{\lambda\})$ in Eq.~\eqref{generating function defenition}, one can insert the explicit matrix elements in the computational basis (cf.\ Eq.~\eqref{eq:matrix_element-computational})
 \begin{equation}\label{eq:gpfaffian}
    g(\{\lambda\})
    =\sum_{\mathcal{J},\mathcal{I}}
    (-1)^{\frac{|\mathcal{I}_1|(|\mathcal{I}_1|+2|\mathcal{J}_1|+1)}{2}}
    \,\operatorname{pf}\!\bigl[\bm{\mathcal{A}}_{\mathcal{J}_0\mathcal{I}_0}\bigr]
    \prod_{j\in\mathcal{J}_0\cup\mathcal{I}_0}\lambda_j.
  \end{equation}
  Next, we introduce a set of Grassmann variables $\{\chi_j\}_{j=1}^{2L}$ and use
  the Berezin‐integral representation of a Pfaffian~\cite{Carac2013}
  \begin{equation}
    \operatorname{pf}\!\bigl[\bm{\mathcal{A}}_{\mathcal{J}_0\mathcal{I}_0}\bigr]
    =\epsilon(\mathcal{J}_0,\mathcal{I}_0)\int\! \mathcal{D}\bm{\chi} \prod_{j\in\mathcal{J}_0\cup\mathcal{I}_0}\chi_j\;e^{\tfrac12\,\bm{\chi}^T\,\bm{\mathcal{A}}\,\bm{\chi}}\,.
  \end{equation}
  where $D\bm{\chi}=D\chi_{2L}D\chi_{2L-1}...D\chi_{1}$ and $ \epsilon(\mathcal{J}_0,\mathcal{I}_0)=(-1)^{|{\mathcal{J}_0\cup\mathcal{I}_0}|(|{\mathcal{J}_0\cup\mathcal{I}_0}|-1)/2}(-1)^{\sum_{j\in\mathcal{J}_0\cup\mathcal{I}_0} j}$. Inserting this in Eq.~\eqref{eq:gpfaffian} leads to
  \begin{equation}\label{generating function defenition2}
    g(\{\lambda\})
    =\sum_{\mathcal{J},\mathcal{I}}
    (-1)^{\frac{|\mathcal{I}_1|(|\mathcal{I}_1|+2|\mathcal{J}_1|+1)}{2}}
    \,\epsilon(\mathcal{J}_0,\mathcal{I}_0)\,
    \int\!\mathcal{D}\bm{\chi}\;e^{\tfrac12\,\bm{\chi}^T\bm{\mathcal{A}}\,\bm{\chi}}
    \prod_{j\in\mathcal{J}_0\cup\mathcal{I}_0}\lambda_j\,\chi_j.
 \end{equation}
  After substituting equation (\ref{sigmaprimeGeneral}) in Eq.~\eqref{generating function defenition2} and making the full exponential Berezin integral over Grassmann variables one can finally obtain Eq.~\eqref{generating function defenition final}.
\end{proof}
\section{Pfaffian Form of Gaussian Operators in Arbitrary Local Pauli Bases}\label{sec:arbitrary-bases}
We now generalize Theorem 3, for which we obtained a compact pfaffian formula in the $\sigma^z$ computational basis, to an arbitrary product Pauli basis.  Let each local basis on site \(j\) be specified by angles \((\phi_j,\theta_j,\alpha_j)\). We label the basis states on \(L\) qubits as
\begin{equation}
\bra{\mathcal{S}} = \bra{s_L, \ldots, s_2, s_1} \quad \text{and} \quad \ket{\mathcal{S}'} = \ket{s'_{1}, s'_{2}, \ldots, s'_{L}}\equiv \ket{s_{L+1}, s_{L+2}, \ldots, s_{2L}},\nonumber
\end{equation}
where \(s_j\) and \(s'_j\) are expressed in the \((\phi_j,\theta_j,\alpha_j)\) basis. As before, \(\mathcal{S}^+\) and \(\mathcal{S}'^+\) denote the sets of qubits with spin up, and \(\mathcal{S}^-\) and \(\mathcal{S}'^-\) denote the sets of qubits with spin down, respectively, and we associate the value \(\bar{s} = +1\) and \(\bar{s}' = +1\) to spin up, and \(\bar{s} = -1\) and \(\bar{s}' = -1\) to spin down. We would like to write an explicit formula for $\bra{\mathcal{S}} \mathcal{G}_{\bm{\mathcal{M}}} \ket{\mathcal{S'}}_{(\boldsymbol{\phi},\boldsymbol{\theta},\boldsymbol{\alpha})}  $. The following theorem provides an efficient and explicit formula to calculate these elements.

\begin{theorem} [ Gaussian Operators in an arbitrary Pauli Basis]\label{theorem 4}
A generic element of a Gaussian operator in an arbitrary Pauli bases can be written as:
\begin{equation}\label{main formula}
 \bra{\mathcal{S}} \mathcal{G}_{\bm{\mathcal{M}}} \ket{\mathcal{S'}}_{(\boldsymbol{\phi},\boldsymbol{\theta},\boldsymbol{\alpha})} = e^{-i\left(\sum_{j\in \mathcal{S}^-}\alpha_j - \sum_{k\in \mathcal{S'}^-}\alpha_k\right)}  \bold{pf}\Bigl[\bm{\mathcal{K}}^{(\boldsymbol{\phi},\boldsymbol{\theta},\boldsymbol{\alpha})}(\mathcal{S},\mathcal{S}')\Bigr].
\end{equation}
For indices \(n > m\), the elements of the antisymmetric matrix \(\bm{\mathcal{K}}^{(\boldsymbol{\phi},\boldsymbol{\theta},\boldsymbol{\alpha})}(\mathcal{S},\mathcal{S}')\) can be found as:
\begin{align}
\mathcal{K}_{mn}^{(\boldsymbol{\phi},\boldsymbol{\theta},\boldsymbol{\alpha})} =\; &
\boldsymbol{\Sigma}_{m,n}\,\mathcal{A}_{mn}\, e^{i(\bar{\phi}_m+\bar{\phi}_n)}
\left(\cos\frac{\theta_m}{2}\right)^{\frac{1+\bar{s}_m}{2}}
\left(\cos\frac{\theta_n}{2}\right)^{\frac{1+\bar{s}_n}{2}}
\left(\sin\frac{\theta_m}{2}\right)^{\frac{1-\bar{s}_m}{2}}
\left(\sin\frac{\theta_n}{2}\right)^{\frac{1-\bar{s}_n}{2}} \nonumber \\[1mm]
&+\, (-1)^{\frac{|\bar{s}_m-\bar{s}_n|}{2}}\, \boldsymbol{\Sigma}'_{m,n}
\left(\sin\frac{\theta_m}{2}\right)^{\frac{1+\bar{s}_m}{2}}
\left(\sin\frac{\theta_n}{2}\right)^{\frac{1+\bar{s}_n}{2}}
\left(\cos\frac{\theta_m}{2}\right)^{\frac{1-\bar{s}_m}{2}}
\left(\cos\frac{\theta_n}{2}\right)^{\frac{1-\bar{s}_n}{2}},
\end{align}
where the phases \(\bar{\phi}_j\) are given by
\begin{equation}
\bar{\phi}_j = 
\begin{cases}
    \phi_j, & j\in \{1,\ldots,L\},\\[1mm]
    -\phi_j, & j\in \{L+1,\ldots,2L\},
\end{cases}
\end{equation}
and the indices for the parameters \(\boldsymbol{\phi}\), \(\boldsymbol{\theta}\), and \(\boldsymbol{\alpha}\) follow a cyclic pattern such that \(L+1 \equiv 1\), \(L+2 \equiv 2\), \(\ldots\), \(2L \equiv L\).

\end{theorem}

\begin{proof}[Proof]
We begin with the fact that, in the rotated Pauli frame, the Gaussian operator can be written as
\begin{equation}\label{gaussian operator alternative basis}
\mathcal G_{\bm{\mathcal M}}^{(\boldsymbol\phi,\boldsymbol\theta,\boldsymbol\alpha)}=U(\phi,\theta,\alpha)
\mathcal G_{\bm{\mathcal M}}^z\;
U^\dagger(\phi,\theta,\alpha)\,,
\end{equation}
where $U(\phi,\theta,\alpha)$ is defined in Eq.\eqref{eq:U}. Hence, the generating function is exactly the same as in Eq.~\eqref{generating function defenition}, except that each matrix element now carries the overall phase and trigonometric weights induced by the conjugation with \(U\).  In general, one can follow the same steps as before to arrive at the compact pfaffian formula in Eq.~\eqref{main formula}. In practice, one proceeds by treating the angle‐dependence in three successive steps.

\paragraph{Step 1. $\boldsymbol{\phi}$ and $\boldsymbol{\alpha}$ angles:}
The contribution of these angles are easy to be seen by directly writing the equation (\ref{gaussian operator alternative basis}). From now on to avoid unnecessary  complication in notation we can put them zero and retrieve them later.

\paragraph{Step 2. $\boldsymbol{\theta}$  angle:}The main idea of finding the contribution of $\theta$ angle  is by looking at the coefficient of $\langle \mathcal{J}|\mathcal{G}_{\bm{\mathcal{M}}}|\mathcal{I} \rangle$, (i.e. which is the element in the $\sigma^z$ basis),  in the expansion of each element of the Gaussian operator matrix in alternative basis, i.e. $\bra{\mathcal{S}} \mathcal{G}_{\bm{\mathcal{M}}} \ket{\mathcal{S'}}_{(0,\boldsymbol{\theta},0)}$. A little inspection shows that the coefficient is the following:
 \begin{equation}\label{trigonometric}
  \prod_{j\in \mathcal{C}}\cos\frac{\theta_j}{2}\prod_{j\in \overline{\mathcal{C}}}\sin\frac{\theta_j}{2},
\end{equation}
where $\mathcal{C}$ is defined as follows: We first define the set of all qubits that the spins in both $\sigma^z$ and $(0,\boldsymbol{\theta},0)$ bases are in the up(down) direction as follows:
 \begin{eqnarray}\label{updown union}
 \mathcal{C}^+&=&(\mathcal{S}^+\cup \mathcal{S}'^+)\cap(\mathcal{J}_1\cup\mathcal{I}_1),\\
 \mathcal{C}^-&=&(\mathcal{S}^-\cup \mathcal{S}'^-)\cap(\mathcal{J}_0\cup\mathcal{I}_0).
\end{eqnarray}
Then we have,
\begin{eqnarray}\label{C C bar}
 \mathcal{C}&=&\mathcal{C}^+\cup\mathcal{C}^-,\\
 \bar{\mathcal{C}}&=&\mathcal{C}^+\cap\mathcal{C}^-.
\end{eqnarray}
In other words $\mathcal{C}$ is the set of all sites that spins in the $\sigma^z$ basis (considering $\langle \mathcal{J}|\mathcal{G}_{\bm{\mathcal{M}}}|\mathcal{I} \rangle$ ) and $(0,\boldsymbol{\theta},0)$ basis are the same.
After factoring out $\prod_{j\in \{\mathcal{S}^+ \ \text{or}  \ \mathcal{{S}}^{'+}\}}\cos\frac{\theta_j}{2} \prod_{j\in \{\mathcal{S}^- \ \text{or}  \ \mathcal{{S}}^{'-}\}} \sin\frac{\theta_j}{2}$ the rest of the expansion (inspired by the generating function defined in the previous section) can be written as:

\begin{equation}\label{elements of density matrix with respect to R matrix}
  \bold{pf}\Bigl[\bm{\mathcal{K}}^{(0,\boldsymbol{\theta},0)}(\mathcal{S},\mathcal{S}')\Bigr],
\end{equation}
where 
\begin{equation}
\begin{split}
 \mathcal{K}^{(0,\boldsymbol{\theta},0)}(\mathcal{S},\mathcal{S}') = \Sigma_{m,n}\,\mathcal{A}_{mn}\, + (-1)^{\frac{|\tilde{s}_m-\tilde{s}_n|}{2}}\, \Sigma_{m,n}' \tan^{\tilde{s}_m}{\frac{\theta_m}{2}}\tan^{\tilde{s}_n}\frac{\theta_n}{2}.
 \end{split}
\end{equation}
Here we take  $\lambda_j=\tan\frac{\theta_j}{2}$($\lambda_j=-\cot\frac{\theta_j}{2}$) when the spins are up (down).

\paragraph{Step 3. Restoring all the angles:}
After absorbing $\prod_{j\in \{\mathcal{S}^+ \ \text{or}  \ \mathcal{{S}}^{'+}\}}\cos\frac{\theta_j}{2} \prod_{j\in \{\mathcal{S}^- \ \text{or}  \ \mathcal{{S}}^{'-}\}} \sin\frac{\theta_j}{2}$ inside the pfaffian and 
restoring all the angles $\boldsymbol{\alpha}$  and $\boldsymbol{\phi}$ we reach to the main formula (\ref{main formula}).

\end{proof}

\subsection{Especial Cases}\label{sec:especial-bases}
In Table \ref{M_nm in different basis} we list several specializations of the general pfaffian kernel obtained by expressing the bra \(\bra{\mathcal S}\) in the \(\mu\)\nobreakdash–basis and the ket \(\ket{\mathcal S'}\) in the \(\nu\)\nobreakdash–basis, with \(\mu,\nu\in\{x,y,z\}\), defined by \(\mathcal K_{mn}^{(\mu\nu})\). Each entry follows by substituting into Theorem 4 the angles
\begin{equation}
(\phi_j,\theta_j,\alpha_j)=
\begin{cases}
(0,0,0), & j\text{ in the }z\text{–basis},\\
(0,\tfrac\pi2,0), & j\text{ in the }x\text{–basis},\\
(\tfrac\pi2,\tfrac\pi2,0), & j\text{ in the }y\text{–basis}.
\end{cases}
\end{equation}
 \begin{table}[h]
\centering
\renewcommand{\arraystretch}{1.5}
\rowcolors{2}{gray!10}{white} 
\begin{tabular}{cl}
\toprule
\rowcolor{gray!30} \textbf{Matrix} & \textbf{Expression} \\
\midrule
$\mathcal{K}_{mn}^{(z,z)}$ & $(\frac{1+\bar{s}_m}{2})(\frac{1+\bar{s}_n}{2})\Sigma_{m,n}\mathcal{A}_{mn}+(\frac{1-\bar{s}_m}{2})(\frac{1-\bar{s}_n}{2})\Sigma_{m,n}'$ \\
$\mathcal{K}_{mn}^{(x,x)}$ & $\frac{1}{2}\left(\Sigma_{m,n}\mathcal{A}_{mn}+\bar{s}_m\bar{s}_n \Sigma_{m,n}'\right)$ \\
$\mathcal{K}_{mn}^{(y,y)}$ & $\frac{1}{2}\left(-\Sigma_{m,n}\mathcal{A}_{mn}+\bar{s}_m\bar{s}_n \Sigma_{m,n}'\right)$ \\
$\mathcal{K}_{mn}^{(z,x)}$ & $\frac{1}{\sqrt{2}}\left[(\frac{1+\bar{s}_m}{2})\Sigma_{m,n}\mathcal{A}_{mn}-\bar{s}_n(\frac{1-\bar{s}_m}{2})\Sigma_{m,n}'\right]$ \\
$\mathcal{K}_{mn}^{(x,z)}$ & $\frac{1}{\sqrt{2}}\left[(\frac{1+\bar{s}_n}{2})\Sigma_{m,n}\mathcal{A}_{mn}-\bar{s}_m(\frac{1-\bar{s}_n}{2})\Sigma_{m,n}'\right]$\\
$\mathcal{K}_{mn}^{(z,y)}$ & $\frac{1}{\sqrt{2}}\left[(\frac{1+\bar{s}_m}{2})i \Sigma_{m,n}\mathcal{A}_{mn}-\bar{s}_n(\frac{1-\bar{s}_m}{2})\Sigma_{m,n}'\right]$ \\
$\mathcal{K}_{mn}^{(y,z)}$ & $\frac{1}{\sqrt{2}}\left[(\frac{1+\bar{s}_n}{2})i \Sigma_{m,n}\mathcal{A}_{mn}-\bar{s}_m(\frac{1-\bar{s}_n}{2})\Sigma_{m,n}'\right]$ \\
$\mathcal{K}_{mn}^{(x,y)},$ $\mathcal{K}_{mn}^{(y,x)}$ & $\frac{1}{2}\left(i \Sigma_{m,n}\mathcal{A}_{mn}+\bar{s}_m\bar{s}_n \Sigma_{m,n}' \right)$ \\
\bottomrule
\end{tabular}
\caption{$\mathcal{K}_{mn}(\boldsymbol{\phi},\boldsymbol{\theta},\boldsymbol{\alpha})$ special cases.}
\label{M_nm in different basis}
\end{table} 
\section{Clifford and Lie Algebraic Structure of $\bold{\Sigma}$ and 
$\bold{\Sigma}'$ Matrices }\label{sec:Sigma-matrices}

In section (\ref{sec:sigmaz-basis}) we proved that there are $2^{2L-1}$ pairs of  the matrices $\boldsymbol{\Sigma}$ and $\boldsymbol{\Sigma}'$. In fact, starting from one pair one can systematically  produce the rest by just using the properties of the pfaffian of the matrices.  
For example, if we define
\begin{equation}
\tilde{\bm{\mathcal{K}}}^{(\boldsymbol{\phi},\boldsymbol{\theta},\boldsymbol{\alpha})}(\mathcal{S},\mathcal{S}')
= P_s \, \bm{\mathcal{K}}^{(\boldsymbol{\phi},\boldsymbol{\theta},\boldsymbol{\alpha})}(\mathcal{S},\mathcal{S}') \, P_s,Elements
\end{equation}
with
\begin{equation}
P_s = \operatorname{diag}\bigl(p_1, p_2, \ldots, p_{2L}\bigr) \quad \text{and} \quad p_j \in \{+1, -1\},
\end{equation}
then it follows that
\begin{equation}
\text{pf}\Bigl[\tilde{\bm{\mathcal{K}}}^{(\boldsymbol{\phi},\boldsymbol{\theta},\boldsymbol{\alpha})}(\mathcal{S},\mathcal{S}')\Bigr] = \pm\text{pf}\Bigl[\bm{\mathcal{K}}^{(\boldsymbol{\phi},\boldsymbol{\theta},\boldsymbol{\alpha})}(\mathcal{S},\mathcal{S}')\Bigr ].
\end{equation}
This invariance implies that while the pfaffian of the matrix remains unchanged modulo a sign, the specific forms of the matrices $\boldsymbol{\Sigma}$ and $\boldsymbol{\Sigma}'$ 
can be modified by such a transformation. By choosing different configurations for the diagonal entries of \(P_s\), one effectively obtains alternative representations for $\boldsymbol{\Sigma}$ and $\boldsymbol{\Sigma}'$.
 
For a system of size \(L\), there are \(2^{2L}\) possible choices for the sign vector \((p_1, p_2, \ldots, p_{2L})\). Since flipping all signs simultaneously leaves the Pfaffian unchanged , only $2^{2L-1}$ such transformations yield genuinely distinct configurations. 
  One can start by considering the configuration where all \(p_j = -1\). It may occur that some of these configurations lead to a negative pfaffian value. In order to rectify this sign ambiguity, one can apply the following adjustments:
\begin{itemize}
    \item If \(L\) is odd, multiply the entire configuration by \(-1\), thereby flipping the sign of the Pfaffian.
    \item If \(L\) is even, multiply all elements of the transformed matrix  by \(-1\), except for those in the first row and first column.
\end{itemize}
In this way one can derive all the possible pairs of the $\boldsymbol{\Sigma}$ and $\boldsymbol{\Sigma}'$ matrices.
\subsection{Lie Algebra Structure}

Beyond the properties mentioned above, we found that these matrices belong to a closed Lie algebra. This algebra is isomorphic to the \(\mathfrak{so}(2L)\) algebra, which is defined as follows:
\begin{equation}
    [X_{ij}, X_{kl}] = i (\delta_{jk} X_{il} - \delta_{ik} X_{jl} + \delta_{il} X_{jk} - \delta_{jl} X_{ik}),
\end{equation}
where $2L\geq i,j,k,l\geq 1$ and $(X_{nm})_{ts}=\delta_{nt}\delta_{ms}-\delta_{ns}\delta_{mt}$.
We have the following theorem: 

\begin{theorem} [ $\boldsymbol{\Sigma}$ and $\boldsymbol{\Sigma}'$ matrices and the \(\mathfrak{so}(2L)\) ]\label{conjecture1}
The matrices $\boldsymbol{\Sigma}$ and $\boldsymbol{\Sigma}'$ are part of an algebra of $L(2L-1)$ generators that are isomorphic to \(\mathfrak{so}(2L)\) algebra.  

\end{theorem}
The proof is presented in the Appendix (\ref{sec:AppendixB}). For example, for  $L=2$, the $\boldsymbol{\Sigma}$ and $\boldsymbol{\Sigma}'$ matrices can be written with respect to the generators of the \(\mathfrak{so}(4)\) algebra as follows:
\begin{equation}
\begin{split}
\boldsymbol{\Sigma}=X_{12}+X_{13}+X_{14}-X_{23}-X_{24}-X_{34},\\
\boldsymbol{\Sigma}'=X_{12}+X_{13}-X_{14}+X_{23}-X_{24}+X_{34}.\\
\end{split}
\end{equation}
The six operators $\boldsymbol{\S}, \boldsymbol{\S}', \boldsymbol{\S}_3, \boldsymbol{\S}_4, \boldsymbol{\S}_5, \boldsymbol{\S}_6$ satisfy the following closed algebra:

\begin{align*}
[\boldsymbol{\S}, \boldsymbol{\S}'] &= \boldsymbol{\S}_3, &
[\boldsymbol{\S}, \boldsymbol{\S}_3] &= \boldsymbol{\S}_4, \\
[\boldsymbol{\S}', \boldsymbol{\S}_3] &= \boldsymbol{\S}_5, &
[\boldsymbol{\S}, \boldsymbol{\S}_4] &= \boldsymbol{\S}_6, \\
[\boldsymbol{\S}', \boldsymbol{\S}_4] &= -12 \boldsymbol{\S}_3 - 2 \boldsymbol{\S}_6, &
[\boldsymbol{\S}_3, \boldsymbol{\S}_4] &= 16 \boldsymbol{\S}, \\
[\boldsymbol{\S}, \boldsymbol{\S}_5] &= -12 \boldsymbol{\S}_3 - 2 \boldsymbol{\S}_6, &
[\boldsymbol{\S}', \boldsymbol{\S}_5] &= -12 \boldsymbol{\S}_3 - \boldsymbol{\S}_6, \\
[\boldsymbol{\S}_3, \boldsymbol{\S}_5] &= 16 \boldsymbol{\S}', &
[\boldsymbol{\S}_4, \boldsymbol{\S}_5] &= 16 \boldsymbol{\S}_3, \\
[\boldsymbol{\S}, \boldsymbol{\S}_6] &= -\tfrac{32}{5} \boldsymbol{\S} - \tfrac{16}{5} \boldsymbol{\S}' - \tfrac{36}{5} \boldsymbol{\S}_4 + \tfrac{12}{5} \boldsymbol{\S}_5, &
[\boldsymbol{\S}', \boldsymbol{\S}_6] &= -\tfrac{16}{5} \boldsymbol{\S} + \tfrac{32}{5} \boldsymbol{\S}' + \tfrac{12}{5} \boldsymbol{\S}_4 - \tfrac{24}{5} \boldsymbol{\S}_5, \\
[\boldsymbol{\S}_3, \boldsymbol{\S}_6] &= 0, &
[\boldsymbol{\S}_4, \boldsymbol{\S}_6] &= \tfrac{576}{5} \boldsymbol{\S} - \tfrac{192}{5} \boldsymbol{\S}' - \tfrac{32}{5} \boldsymbol{\S}_4 - \tfrac{16}{5} \boldsymbol{\S}_5, \\
[\boldsymbol{\S}_5, \boldsymbol{\S}_6] &= -\tfrac{192}{5} \boldsymbol{\S} + \tfrac{384}{5} \boldsymbol{\S}' - \tfrac{16}{5} \boldsymbol{\S}_4 + \tfrac{32}{5} \boldsymbol{\S}_5.
\end{align*}

\subsection{Clifford Algebra Embedding}

The \(\mathfrak{so}(2L)\) algebra can be embedded into the even part of the Clifford algebra $Cl(2L)$, using the identity 
\begin{equation}
X_{ij}=\frac{1}{4}[\gamma_i,\gamma_j]
\end{equation}
where $\{\gamma_i,\gamma_j\}=2\delta_{ij}\bold{I}$. This helps to write the $\boldsymbol{\Sigma}$ and $\boldsymbol{\Sigma}'$ matrices with respect to the $\gamma$ matrices of the Clifford algebra. We note that if one finds a single representation of the $\gamma$ matrices composed of $+1$, $-1$, and $0$, then a total of $2^{2L-1}$ such representations can be constructed. This arises from the freedom to perform the transformation $\gamma_i \to \epsilon_i \gamma_i$, where $\epsilon_i^2 = 1$. However, flipping the sign of all $\gamma_i$ matrices results in an equivalent representation. Consequently, there are $2^{2L-1}$ distinct choices for the $\gamma$ matrices. This indicates again the possibility of $2^{2L-1}$ distinct choices for the $\boldsymbol{\Sigma}$ and $\boldsymbol{\Sigma}'$ matrices.
\section{Conclusion}\label{sec:conclusion}
In this work, we have developed a fully unified Pfaffian framework for evaluating matrix elements of arbitrary fermionic Gaussian operators—whether pure unitaries, or mixed‐state density matrices, in any product Pauli basis. By introducing a discrete pair of sign–encoding matrices with entries in \(\{\pm1\}\), we have resolved all phase ambiguities inherent in Pfaffian expressions under local spin rotations and reduced every overlap or expectation value to a single Pfaffian of a \(2L\times2L\) kernel. This construction not only subsumes and generalizes every previously known formula in special bases, but also reveals an elegant algebraic structure: under commutation, our sign matrices generate a Lie algebra isomorphic to \(\mathfrak{so}(2L)\). The practical payoff is substantial. Basis‐rotated simulation of free‐fermion dynamics, randomized (shadow) measurement protocols, and even online post‐measurement state updates all collapse from exponentially large Pauli‐string expansions into \(O(L^3)\) linear‐algebra routines.

Looking forward, this Pfaffian machinery opens the door to several concrete advances. One can now fit Gaussian covariances directly from randomized‐measurement data via maximum‐likelihood or Bayesian inference—each shot contributing one Pfaffian term—bypassing the need for multi‐basis fixed tomography. Real‐time quench and transport calculations reduce to orthogonal rotations of a covariance followed by Pfaffians, enabling simulations of entanglement growth and Loschmidt echoes at scales previously out of reach. Moreover, weak non‐Gaussian perturbations or error‐mitigation channels can be embedded into this Gaussian core as low‐rank updates to the Pfaffian kernel, preserving \(O(L^3)\) efficiency. In all these applications, the  complexity of free‐fermion systems gives way to the manipulation of a single \(2L\times2L\) matrix and modest Pfaffian evaluations, offering both analytic transparency and scalable performance across quantum information and many‐body physics.
\section*{Acknowledgements}
 We thank M. Heyl for reading the manuscript and
comments. MAR thanks CNPq and FAPERJ (grant number E-26/210.062/2023) for partial support.

\newpage

\appendix
\section*{\centering \Huge \textbf{Appendix}}
\addcontentsline{toc}{section}{Appendix}
\vspace*{2em}  
\addtocontents{toc}{\protect\setcounter{tocdepth}{1}}
\renewcommand{\theequation}{\thesection.\arabic{equation}}

\setcounter{equation}{0}
\setcounter{table}{0}
\renewcommand{\thetable}{A\arabic{table}}

\section{Examples of the $\Sigma$ and $\Sigma'$ matrices }\label{sec:AppendixA}
In this Appendix we provide explicit forms of the $\boldsymbol{\Sigma}$ and $\boldsymbol{\Sigma}'$ matrices for $L=2$ and $3$. Half of the list of  the possible pairs for $L=2$ are in the Table \ref{Table:L2}, the other half can be found by just taking the negative of the presented matrices.

For $L=3$ there are $32$ possible pairs. We show in the Table \ref{Table:L3} four of them. All the possible cases can be found using the procedure outlined in the main text.

\begin{table}[ht]\label{Table:L2}
\centering
\caption{Four explicit pairs $(\boldsymbol{\Sigma},\boldsymbol{\Sigma}')$ for $L=2$. Pairs 5--8 are obtained by taking the overall negative of pairs 1--4, respectively.}
\begin{tabular}{c c c}
\toprule
Pair & $\boldsymbol{\Sigma}$ & $\boldsymbol{\Sigma}'$ \\
\midrule
1 & 
$\begin{pmatrix}
0 & 1 & 1 & -1\\[4mm]
-1 & 0 & 1 & -1\\[4mm]
-1 & -1 & 0 & -1\\[4mm]
1 & 1 & 1 & 0
\end{pmatrix}$ &
$\begin{pmatrix}
0 & 1 & 1 & 1\\[4mm]
-1 & 0 & -1 & -1\\[4mm]
-1 & 1 & 0 & 1\\[4mm]
-1 & 1 & -1 & 0
\end{pmatrix}$ \\[6mm]
\midrule
2 & 
$\begin{pmatrix}
0 & 1 & -1 & 1\\[4mm]
-1 & 0 & -1 & 1\\[4mm]
1 & 1 & 0 & -1\\[4mm]
-1 & -1 & 1 & 0
\end{pmatrix}$ &
$\begin{pmatrix}
0 & 1 & -1 & -1\\[4mm]
-1 & 0 & 1 & 1\\[4mm]
1 & -1 & 0 & 1\\[4mm]
1 & -1 & -1 & 0
\end{pmatrix}$ \\[6mm]
\midrule
3 & 
$\begin{pmatrix}
0 & -1 & 1 & 1\\[4mm]
1 & 0 & -1 & -1\\[4mm]
-1 & 1 & 0 & 1\\[4mm]
-1 & 1 & -1 & 0
\end{pmatrix}$ &
$\begin{pmatrix}
0 & -1 & 1 & -1\\[4mm]
1 & 0 & 1 & -1\\[4mm]
-1 & -1 & 0 & -1\\[4mm]
1 & 1 & 1 & 0
\end{pmatrix}$ \\[6mm]
\midrule
4 & 
$\begin{pmatrix}
0 & -1 & -1 & -1\\[4mm]
1 & 0 & 1 & 1\\[4mm]
1 & -1 & 0 & 1\\[4mm]
1 & -1 & -1 & 0
\end{pmatrix}$ &
$\begin{pmatrix}
0 & -1 & -1 & 1\\[4mm]
1 & 0 & -1 & 1\\[4mm]
1 & 1 & 0 & -1\\[4mm]
-1 & -1 & 1 & 0
\end{pmatrix}$ \\
\bottomrule
\end{tabular}
\label{tab:SigmaPairs}
\end{table}

\begin{table}[htbp]\label{Table:L3}
\centering
\caption{Four representative pairs $(\boldsymbol{\Sigma},\boldsymbol{\Sigma}')$ for $L=3$.
The other 28 pairs are obtained by overall sign changes and are algebraically equivalent.}
\begin{tabular}{c c c}
\toprule
Pair & $\boldsymbol{\Sigma}$ & $\boldsymbol{\Sigma}'$ \\
\midrule
1 & 
$\displaystyle
\begin{pmatrix}
0 & 1 & 1 & 1 & 1 & -1\\[2mm]
-1 & 0 & 1 & 1 & 1 & -1\\[2mm]
-1 & -1 & 0 & 1 & 1 & -1\\[2mm]
-1 & -1 & -1 & 0 & 1 & -1\\[2mm]
-1 & -1 & -1 & -1 & 0 & -1\\[2mm]
1 & 1 & 1 & 1 & 1 & 0
\end{pmatrix}
$ 
&
$\displaystyle
\begin{pmatrix}
0 & 1 & -1 & 1 & -1 & -1\\[2mm]
-1 & 0 & 1 & -1 & 1 & 1\\[2mm]
1 & -1 & 0 & 1 & -1 & -1\\[2mm]
-1 & 1 & -1 & 0 & -1 & -1\\[2mm]
1 & -1 & 1 & 1 & 0 & 1\\[2mm]
1 & -1 & 1 & 1 & -1 & 0
\end{pmatrix}
$ \\[4mm]
\midrule
2 & 
$\displaystyle
\begin{pmatrix}
0 & 1 & 1 & -1 & -1 & -1\\[2mm]
-1 & 0 & 1 & -1 & -1 & -1\\[2mm]
-1 & -1 & 0 & -1 & -1 & -1\\[2mm]
1 & 1 & 1 & 0 & 1 & 1\\[2mm]
1 & 1 & 1 & -1 & 0 & 1\\[2mm]
1 & 1 & 1 & -1 & -1 & 0
\end{pmatrix}
$ 
&
$\displaystyle
\begin{pmatrix}
0 & 1 & -1 & -1 & 1 & -1\\[2mm]
-1 & 0 & 1 & 1 & -1 & 1\\[2mm]
1 & -1 & 0 & -1 & 1 & -1\\[2mm]
1 & -1 & 1 & 0 & -1 & 1\\[2mm]
-1 & 1 & -1 & 1 & 0 & -1\\[2mm]
1 & -1 & 1 & -1 & 1 & 0
\end{pmatrix}
$ \\[4mm]
\midrule
3 & 
$\displaystyle
\begin{pmatrix}
0 & -1 & 1 & -1 & 1 & -1\\[2mm]
1 & 0 & -1 & 1 & -1 & 1\\[2mm]
-1 & 1 & 0 & -1 & 1 & -1\\[2mm]
1 & -1 & 1 & 0 & -1 & 1\\[2mm]
-1 & 1 & -1 & 1 & 0 & -1\\[2mm]
1 & -1 & 1 & -1 & 1 & 0
\end{pmatrix}
$ 
&
$\displaystyle
\begin{pmatrix}
0 & -1 & -1 & -1 & -1 & -1\\[2mm]
1 & 0 & -1 & -1 & -1 & -1\\[2mm]
1 & 1 & 0 & -1 & -1 & -1\\[2mm]
1 & 1 & 1 & 0 & 1 & 1\\[2mm]
1 & 1 & 1 & -1 & 0 & 1\\[2mm]
1 & 1 & 1 & -1 & -1 & 0
\end{pmatrix}
$ \\[4mm]
\midrule
4 & 
$\displaystyle
\begin{pmatrix}
0 & 1 & -1 & -1 & -1 & 1\\[2mm]
-1 & 0 & -1 & -1 & -1 & 1\\[2mm]
1 & 1 & 0 & 1 & 1 & -1\\[2mm]
1 & 1 & -1 & 0 & 1 & -1\\[2mm]
1 & 1 & -1 & -1 & 0 & -1\\[2mm]
-1 & -1 & 1 & -1 & 1 & 0
\end{pmatrix}
$ 
&
$\displaystyle
\begin{pmatrix}
0 & 1 & 1 & -1 & 1 & 1\\[2mm]
-1 & 0 & -1 & 1 & -1 & -1\\[2mm]
-1 & 1 & 0 & 1 & -1 & -1\\[2mm]
1 & -1 & -1 & 0 & -1 & -1\\[2mm]
-1 & 1 & 1 & 1 & 0 & 1\\[2mm]
-1 & 1 & 1 & 1 & -1 & 0
\end{pmatrix}
$ \\
\bottomrule
\end{tabular}
\label{tab:SigmaPairsL3}
\end{table}
\newpage
\section{Proof of Theorem 5: Lie Algebra Generation by \texorpdfstring{\(\boldsymbol{\Sigma}\)}{Sigma} and \texorpdfstring{\(\boldsymbol{\Sigma}'\)}{Sigma'}}\label{sec:AppendixB}

In this appendix, we provide a detailed proof of Theorem 5.  
We aim to prove that, for any integer \(L \geq 3\), the two explicitly constructed antisymmetric matrices  
\(\boldsymbol{\Sigma}, \boldsymbol{\Sigma}' \in \mathfrak{so}(2L)\)  
generate the entire Lie algebra \(\mathfrak{so}(2L)\) under successive commutators. That is, starting from \(\boldsymbol{\Sigma}\) and \(\boldsymbol{\Sigma}'\), and repeatedly taking commutators, one obtains a vector space of antisymmetric matrices of dimension \(L(2L - 1)\).

\subsection{Definitions}

This section introduces the key definitions and constructs used in the analysis, including the sign vector and two associated antisymmetric matrices. We now define the matrices \(\boldsymbol{\Sigma}\) and \(\boldsymbol{\Sigma}'\), both of which are real and antisymmetric.
\begin{itemize}

  \item \textbf{Matrix \(\boldsymbol{\Sigma}\)} 
  \begin{equation}
    \Sigma_{ij} = \begin{cases} 
      p_i p_j & i < j, \\ 
      -p_j p_i & i > j, \\ 
      0 & i = j. 
    \end{cases}
  \end{equation}
  where \(p = (p_1, \dots, p_{2L})\) with each \(p_i \in \{+1, -1\}\).
  
  \item \textbf{Matrix \(\boldsymbol{\Sigma}'\)}
  \begin{equation}
    \Sigma'_{ij} = \begin{cases} 
      \operatorname{sgn}(i,j) p_i p_j & i < j, \\ 
      -\operatorname{sgn}(j,i) p_j p_i & i > j, \\ 
      0 & i = j. 
    \end{cases}
  \end{equation}
To define the sign function $\operatorname{sgn}(i,j)$, we first introduce the function \( f(i,j) \), which depends on the parity of \( L \). For odd \( L \), we define
\begin{equation}
f(i,j)=
\begin{cases}
i+j+1, & (i\le L\ \&\ j\le L)\;\text{or}\;(i\le L<j),\\[6pt]
i+j+2, & i>L\ \text{and}\ j> L .
\end{cases}
\end{equation}
while for even \( L \), we define
\begin{equation}
f(i,j)=
\begin{cases}
i+j+1, & (i\le L\ \&\ j\le L),\\[6pt]
i+j+2, & i>L\ \text{and}\ j> L \;\text{or}\;(i\le L<j).
\end{cases}
\end{equation}
Using this, we define the sign function
\begin{equation}\label{eq:sgn}
\operatorname{sgn}(i,j) = (-1)^{f(i,j)}, \qquad 1 \le i,j \le 2L~.
\end{equation}
  \item \textbf{Matrices \(\tilde{\boldsymbol{\Sigma}}\) and \(\tilde{\boldsymbol{\Sigma}}'\)
  }
  \begin{equation}
    {\boldsymbol{\Sigma}} = \boldsymbol{g}\,\tilde{\boldsymbol{\Sigma}}\,\boldsymbol{g}^{T},
    \qquad
    {\boldsymbol{\Sigma}}' = \boldsymbol{g}\,\tilde{\boldsymbol{\Sigma}}'\,\boldsymbol{g}^{T},
  \end{equation}
  where \(\boldsymbol{g} = \boldsymbol{P}\boldsymbol{g}_0 \) is an orthogonal matrix with
  \begin{equation}
  \boldsymbol{g}_0 = 
  \bigl[
    \phi^{(1)}, \psi^{(1)}, \dots, \phi^{(L)}, \psi^{(L)}
  \bigr], \quad
  \boldsymbol{P} = \operatorname{diag}(p_1, \dots, p_{2L}),
  \end{equation}
  and
  \begin{equation}
  \phi^{(k)}_j = \sqrt{\tfrac{1}{L}} \sin\!\left(\tfrac{(j-1)(2k-1)\pi}{2L}\right),
  \quad
  \psi^{(k)}_j = \sqrt{\tfrac{1}{L}} \cos\!\left(\tfrac{(j-1)(2k-1)\pi}{2L}\right).
  \end{equation}
\end{itemize}

\subsection{A Theorem}
In this subsection, we use the Two-Element Generation Criterion Theorem ~\cite{Kuranishi1951,Humphreys1972} to prove that any two suitably chosen antisymmetric matrices  
\(\boldsymbol{\Sigma}, \boldsymbol{\Sigma}' \in \mathfrak{so}(2L)\)  
generate the entire Lie algebra \(\mathfrak{so}(2L)\) under successive commutators.

Let \(\mathfrak{g}\) be a simple complex Lie algebra of rank \(r\). Fix a Cartan subalgebra \(\mathfrak{h} \subset \mathfrak{g}\), with root decomposition
\begin{equation}
  \mathfrak{g} = \mathfrak{h} \oplus \bigoplus_{\alpha \in \Delta} \mathfrak{g}_\alpha,
\end{equation}
where each \(\mathfrak{g}_\alpha\) is the one-dimensional root space corresponding to root \(\alpha\).
%
If the following conditions hold:

\begin{enumerate}
  \item \textbf{Regular semisimplicity:} An element \(\boldsymbol{H} \in \mathfrak{h}\) is called regular if \(\alpha(\boldsymbol{H}) \neq 0\) for all roots \(\alpha \in \Delta\). Equivalently, the centralizer of \(\boldsymbol{H}\) is exactly \(\mathfrak{h}\).

  \item \textbf{Extra-root condition:} An element \(\boldsymbol{E} \in \mathfrak{g}\) has a nonzero projection onto each of the \(r\) simple-root spaces \(\mathfrak{g}_{\alpha_i}\) for the chosen simple-root system \(\{\alpha_1, \dots, \alpha_r\}\).
\end{enumerate}

Then the Lie subalgebra generated by \(\boldsymbol{H}\) and \(\boldsymbol{E}\), i.e.,
\begin{equation}
  \mathrm{Lie}\{\boldsymbol{H}, \boldsymbol{E}\},
\end{equation}
is the full Lie algebra \(\mathfrak{g}\). In particular, it has dimension
\begin{equation}
  \dim \mathfrak{g} = r + |\Delta|.
\end{equation}
A detailed exposition of this result and its applications within the structure theory of semisimple Lie algebras can be found in~\cite{Kuranishi1951,Humphreys1972}.
%
%
To prove that the Lie algebra generated by \(\boldsymbol{\Sigma}\) and \(\boldsymbol{\Sigma}'\) is all of \(\mathfrak{so}(2L)\), we apply the Two-Element Generation Criterion Theorem. For this purpose, we consider the conjugated matrices
\begin{equation}
  \boldsymbol{H} := \tilde{\boldsymbol{\Sigma}} \qquad \text{and} \qquad \boldsymbol{E} := \tilde{\boldsymbol{\Sigma}}',
\end{equation}
and verify that they satisfy the two structural conditions of the theorem within the Lie algebra \(\mathfrak{so}(2L)\).

\subsubsection{\texorpdfstring{\(\boldsymbol{\Sigma}\)}{Sigma} is regular semisimple}
\begin{enumerate}
  \item \textbf{Spectrum:} 
The antisymmetric matrix $\boldsymbol{\Sigma}$ has purely imaginary eigenvalues
\begin{equation}
  \lambda_k = i\,\omega_k,
  \quad
  \lambda_{k+L}=-\lambda_k,
  \quad
  \omega_k = -\cot\!\Bigl(\frac{(2k-1)\,\pi}{4L}\Bigr),
  \quad
  k=1,\dots,L,
\end{equation}
all distinct.  In fact, one can check that
\begin{equation}
  \tilde{\boldsymbol{\Sigma}}=\boldsymbol{g}^{T}\,\boldsymbol{\Sigma}\,\boldsymbol{g}
  \;=\;
  \bigoplus_{k=1}^{L}
  \begin{pmatrix}
    0 & \omega_k \\[4pt]
   -\omega_k & 0
  \end{pmatrix},
\end{equation}
  \item \textbf{Centralizer:} One checks explicitly that
  \begin{equation}
    [\tilde{\boldsymbol{\Sigma}}, \boldsymbol{X}_{2i-1,2i}] = 0 \quad \text{for } i=1,\dots,L,
  \end{equation}
  but \([\tilde{\boldsymbol{\Sigma}}, \boldsymbol{X}_{jk}] \ne 0\) for the rest. Hence
  \begin{equation}
    \text{Cent}(\tilde{\boldsymbol{\Sigma}}) = \mathrm{span}\{\boldsymbol{H}_i = \boldsymbol{X}_{2i-1,2i}\}_{i=1}^L
  \end{equation}
  is exactly a Cartan subalgebra of dimension \(L\).
\end{enumerate}
Together, these confirm that \(\boldsymbol{\Sigma}\) is a regular semisimple element, fulfilling the first condition of the theorem.

\subsubsection{\texorpdfstring{\(\boldsymbol{\Sigma}'\)}{Sigma'} has nonzero simple-root components}
The matrix $\tilde{\boldsymbol{\Sigma}}'$ can be written as

\begin{equation}
\tilde{\Sigma}'_{\mu\nu}
  =\sum_{r=1}^{2L}\sum_{s=1}^{2L} S_{rs}\,
    \begin{cases}
      \phi^{(\frac{\mu + 1}{2})}_r\,\phi^{(\frac{\nu + 1}{2})}_s, & \mu,\nu\text{ odd},\\[8pt]
      \phi^{(\frac{\mu + 1}{2})}_r\,\psi^{(\frac{\nu}{2})}_s,    & \mu\text{ odd},\;\nu\text{ even},\\[8pt]
      \psi^{(\frac{\mu}{2})}_r\,  \phi^{(\frac{\nu + 1}{2})}_s,  & \mu\text{ even},\;\nu\text{ odd},\\[8pt]
      \psi^{(\frac{\mu}{2})}_r\,  \psi^{(\frac{\nu}{2})}_s,     & \mu,\nu\text{ even}.
    \end{cases}
\end{equation}
where the antisymmetric sign matrix \( S_{ij} \) is
\begin{equation}
S_{ij} =
\begin{cases}
\operatorname{sgn}(i,j),  & i < j, \\[4pt]
-\operatorname{sgn}(j,i), & i > j, \\[4pt]
0,                        & i = j.
\end{cases}
\qquad (1 \le i,j \le 2L).
\end{equation}
and $\operatorname{sgn}(i,j)$ is defined in Eq.~\eqref{eq:sgn}.
For the Cartan subalgebra 
\begin{equation}
H_i \;=-i\;X_{\,2i-1,\,2i}\,,\qquad i=1,\dots,L,
\end{equation}
a convenient simple–root system of \(\mathfrak{so}(2L)\) is generated by  
\begin{eqnarray}
E_{\alpha_i} &=&\frac{1}{2}( X_{\,2i-1,\,2i+1} \;+\; X_{\,2i,\,2(i+1)}-iX_{2i-1,2i+2}+iX_{2i,2i+1}),
      \hspace{0.3cm} i = 1,\dots ,L-1,\\[6pt]
E_{\alpha_L} &=&\frac{1}{2}(- X_{\,2L-3,\,2L-1} \;+\; X_{\,2L-2,\,2L}-iX_{2L-3,2L}-iX_{2L-2,2L-1}).
\end{eqnarray}
We aim to show that the matrix $\tilde{\boldsymbol{\Sigma}}'$ has nonzero overlap with all the simple root generators $E_{\alpha_i}$. To this end, we use the Hilbert–Schmidt inner product on the Lie algebra $\mathfrak{so}(2L)$, defined as
\begin{equation}
\langle A, B \rangle = 
\tfrac{1}{2} \, \mathrm{tr}(A^\dagger B)
= \sum_{p<q} A^*_{pq} \, B_{pq},
\qquad A, B \in \mathfrak{so}(2L).
\end{equation}
Before proceeding with the overlap computation, we first prove a useful identity involving the standard basis elements of the Lie algebra. Let $M$ be any complex antisymmetric $2L \times 2L$ matrix, i.e., $M_{ij} = -M_{ji}$ and $M_{ii} = 0$. Consider the elementary generator $X_{a,b}$ defined by
\begin{equation}
(X_{a,b})_{mn} = \delta_{a,m}\delta_{b,n} - \delta_{a,n}\delta_{b,m},
\quad \text{so that} \quad X_{a,b} = - X_{b,a}.
\end{equation}
We now compute the inner product between $X_{a,b}$ and $M$:
\begin{align}
\langle X_{a,b}, M \rangle 
&= \sum_{i<j} (X_{a,b})^*_{ij} \, M_{ij} \notag \\
&= \sum_{i<j} (X_{a,b})_{ij} \, M_{ij}
\quad \text{(since $X_{a,b}$ is real)} \notag \\
&= (X_{a,b})_{a,b} \, M_{a,b} 
+ \sum_{\substack{i<j \\ (i,j) \ne (a,b)}} (X_{a,b})_{ij} \, M_{ij}.
\end{align}
By construction of $X_{a,b}$, we have:
\begin{equation}
(X_{a,b})_{ij} =
\begin{cases}
+1, & \text{if } (i,j) = (a,b),\\
-1, & \text{if } (i,j) = (b,a),\\
0, & \text{otherwise}.
\end{cases}
\end{equation}
However, in the sum over $i < j$, the pair $(b,a)$ does not appear (since $b > a$), so only the $(a,b)$ term contributes:
\begin{equation}
\langle X_{a,b}, M \rangle = M_{ab}.
\end{equation}
Using that one can show for \(l=1,\dots ,L-1\),
\begin{align}
\langle E_{\alpha_l},\,\tilde{\boldsymbol{\Sigma}}' \rangle &=\;
\frac{1}{2}\, \tilde{\Sigma}'_{2l-1,\,2l+1}
+ \frac{1}{2}\, \tilde{\Sigma}'_{2l,\,2l+2}
+ \frac{i}{2}\, \tilde{\Sigma}'_{2l-1,\,2l+2}
- \frac{i}{2}\, \tilde{\Sigma}'_{2l,\,2l+1} \nonumber\\[6pt]
&= \frac{1}{2} \sum_{r,s=1}^{2L} S_{rs} \left[
\phi^{(2l-1)}_r\,\phi^{(2l+1)}_s
+ \psi^{(2l)}_r\,\psi^{(2l+2)}_s
+ i\,\phi^{(2l-1)}_r\,\psi^{(2l+2)}_s
- i\,\psi^{(2l)}_r\,\phi^{(2l+1)}_s
\right]
\end{align}
Now insert the definitions of $\phi^{(k)}_j$ and $\psi^{(k)}_j$:
\begin{align}
\langle E_{\alpha_l},\,\tilde{\boldsymbol{\Sigma}}' \rangle
=& \frac{1}{2L} \sum_{r,s=1}^{2L} S_{rs} \bigg[
 \sin\!\left(\tfrac{(r-1)(2l - 1)\pi}{2L}\right)
  \sin\!\left(\tfrac{(s-1)(2l + 1)\pi}{2L}\right)\nonumber \\[4pt]
+& \cos\!\left(\tfrac{(r-1)(2l - 1)\pi}{2L}\right)
   \cos\!\left(\tfrac{(s-1)(2l + 1)\pi}{2L}\right)+ i \sin\!\left(\tfrac{(r-1)(2l - 1)\pi}{2L}\right)
     \cos\!\left(\tfrac{(s-1)(2l + 1)\pi}{2L}\right)\nonumber \\[4pt]
-& i \cos\!\left(\tfrac{(r-1)(2l - 1)\pi}{2L}\right)
     \sin\!\left(\tfrac{(s-1)(2l + 1)\pi}{2L}\right)
\bigg]=\frac{1}{2L}\sum_{r,s=1}^{2L}e^{-i\tfrac{(s-1)(2l + 1)\pi}{2L}}S_{r,s} e^{i\tfrac{(r-1)(2l - 1)\pi}{2L}}\nonumber\\
=&-\frac{e^{-\frac{i(l-1)\pi }{L}} \left(-1+e^{\frac{2il\pi }{L}}\right)}{\left(-1+e^{\frac{i \pi }{L}}\right) L \left(\cos \left(\frac{\pi }{2 L}\right)+\cos \left(\frac{l \pi }{L}\right)\right)}
\end{align}

and for the last generator
\begin{align}
\langle E_{\alpha_L},\,\tilde{\boldsymbol{\Sigma}}' \rangle &=\;
- \frac{1}{2}\, \tilde{\Sigma}'_{2L-3,\,2L-1}
+ \frac{1}{2}\, \tilde{\Sigma}'_{2L-2,\,2L}
+ \frac{i}{2}\, \tilde{\Sigma}'_{2L-3,\,2L}
+ \frac{i}{2}\, \tilde{\Sigma}'_{2L-2,\,2L-1}\nonumber \\[6pt]
&= \sum_{r,s=1}^{2L} S_{rs} \left[
-\tfrac{1}{2} \phi^{(2L-3)}_r\,\phi^{(2L-1)}_s
+ \tfrac{1}{2} \psi^{(2L-2)}_r\,\psi^{(2L)}_s
+ \tfrac{i}{2} \phi^{(2L-3)}_r\,\psi^{(2L)}_s
+ \tfrac{i}{2} \psi^{(2L-2)}_r\,\phi^{(2L-1)}_s
\right]
\end{align}
By replacing $\phi^{(k)}_j$ and $\psi^{(k)}_j$ one has
\begin{align}
\langle E_{\alpha_L},\,\tilde{\boldsymbol{\Sigma}}' \rangle
& = \frac{1}{2L} \sum_{r,s=1}^{2L} S_{rs} \bigg[
 - \sin\!\left(\tfrac{(r-1)(2L - 3)\pi}{2L}\right)
     \sin\!\left(\tfrac{(s-1)(2L - 1)\pi}{2L}\right)\nonumber \\[4pt]
+& \cos\!\left(\tfrac{(r-1)(2L - 3)\pi}{2L}\right)
     \cos\!\left(\tfrac{(s-1)(2L - 1)\pi}{2L}\right)+ i \sin\!\left(\tfrac{(r-1)(2L - 3)\pi}{2L}\right)
     \cos\!\left(\tfrac{(s-1)(2L - 1)\pi}{2L}\right) \nonumber\\[4pt]
+& i \cos\!\left(\tfrac{(r-1)(2L - 3)\pi}{2L}\right)
     \sin\!\left(\tfrac{(s-1)(2L - 1)\pi}{2L}\right)
\bigg]=\frac{1}{2L} \sum_{r,s=1}^{2L} e^{i\tfrac{(s-1)(2L - 1)\pi}{2L}} S_{rs} e^{i\tfrac{(r-1)(2L - 3)\pi}{2L}}\nonumber\\
=&
\frac{i(-1)^{L}e^{\frac{i \pi }{L}} \csc ^2\left(\frac{\pi }{4 L}\right)}{2L(2 \cos \left(\frac{\pi }{2 L}\right)+1)}
\end{align}
None of the above overlaps can be zero.
\subsubsection{Conclusion}

All hypotheses of the Two-Element Generation Criterion Theorem are satisfied:
\begin{itemize}
  \item \(\mathfrak{so}(2L)\) is simple for \(L\ge3\).
  \item \(\boldsymbol{\Sigma}\) is regular semisimple, with centralizer equal to the Cartan subalgebra of dimension \(L\).
  \item \(\boldsymbol{\Sigma}'\) has nonzero projection onto every simple-root space \(\mathfrak{g}_{\alpha_i}\).
\end{itemize}
It follows that the Lie subalgebra generated by \(\boldsymbol{\Sigma}\) and \(\boldsymbol{\Sigma}'\) under iterated commutators is all of \(\mathfrak{so}(2L)\):
\[
  \mathrm{Lie} \{ \boldsymbol{\Sigma}, \boldsymbol{\Sigma}' \} = \mathfrak{so}(2L).
\]
Therefore, these two matrices generate all \(L(2L-1)\) independent antisymmetric matrices.

\end{document}